\documentclass[prl,nofootinbib,floats,aps,twocolumn,tightenlines,superscriptaddress]{revtex4-1}
\usepackage{amsmath}
\usepackage{amsfonts}
\usepackage{dsfont} 
\usepackage{graphicx}
\usepackage{todonotes}
\usepackage{gensymb}
\usepackage{bbold}
\usepackage{bm}
\usepackage{natbib}
\usepackage{amssymb}
\usepackage{latexsym}
\usepackage{fancyhdr}
\usepackage[bbgreekl]{mathbbol}
\usepackage{epsfig}
\usepackage{color}
\usepackage{textcomp}
\usepackage{amsthm}
\usepackage{ulem}
\usepackage{hyperref}
\usepackage{mathrsfs}
\usepackage{cleveref}
\usepackage[T1]{fontenc}
\usepackage{natbib}

\newtheorem{theorem}{Theorem}[section]

\newtheorem{lemma}[theorem]{Lemma}

\newcommand{\ket}[1]{\left| {#1} \right\rangle}
\newcommand{\bra}[1]{\left\langle {#1} \right|}

\newcommand{\proj}[1]{\mathinner{|{#1}\rangle}\!\!\mathinner{\langle{#1}|}}
\newcommand{\ketbra}[2]{\mathinner{|{#1}\rangle}\!\!\mathinner{\langle{#2}|}}

\newcommand{\tr}[1]{\mbox{$\mathrm{Tr}\left(#1\right)$}}

\newcommand{\Id}{\mathds{1}}

\newcommand{\ii}{\mathrm{i}}

\newcommand{\cB}{\mathcal{B}}

\newcommand{\cE}{\mathcal{E}}
\newcommand{\cF}{\mathcal{F}}
\newcommand{\cH}{\mathcal{H}}
\newcommand{\sA}{\mathscr{A}}
\newcommand{\sB}{\mathscr{B}}
\newcommand{\dC}{\mathds{C}}

\renewcommand{\t}[1]{\mathrm{#1}}
\newcommand{\be}{\begin{equation}}
\newcommand{\ee}{\end{equation}}

\begin{document}

\title{Noise-resistant device-independent certification of Bell state measurements}

\author{Jean-Daniel Bancal}
\author{Nicolas Sangouard}
\author{Pavel Sekatski}
\affiliation{Quantum Optics Theory Group, Universit\"at Basel, Klingelbergstrasse 82, CH-4056 Basel, Switzerland}
\date{\today}

\begin{abstract}
Device-independent certification refers to the characterization of an apparatus without reference to the internal description of other devices. It is a trustworthy certification method, free of assumption on the underlying Hilbert space dimension and on calibration methods. We show how it can be used to quantify the quality of a Bell state measurement, whether deterministic, partial or probabilistic. Our certification is noise resistant and opens the way towards the device-independent self-testing of Bell state measurements in existing experiments.
\end{abstract}

\maketitle

\paragraph{Introduction--} In quantum theory, the relation between the state of a physical system and the information that can be obtained about it upon measurement is non-trivial: It is the combination of the system's state and measurement that determines the probability of measurement outcomes. Thereby, the measurement acquires an irreducible status on equal footing with the system's state.\\

In an actual experiment, an observer can reconstruct these probabilities by repeatedly performing the same measurements on copies of the same state. Yet, even under the assumption that all experimental runs are identical, backtracking the quantum description of the state and measurements from the observed probabilities is not straightforward and does not necessarily lead to a unique solution. This difficulty is usually circumvented by assuming a quantum description either of the measurement, e.g. for state tomography, or of the state, e.g. for detector characterization. Such assumptions typically rely on the best understanding of the physics of the measured system available to day -- a physical model consistent with the history of previous experiments performed on the same or similar setups. Nonetheless, this does not protect against interpretation bias, which can lead to erroneous conclusions, e.g. about the presence of quantum features~\cite{Rosset12} or the security of quantum communication protocols~\cite{Acin06}. In contrast with this conventional approach, there is a program aiming to reconstruct as much as possible of a quantum description of the setup without assumptions about the internal functioning of any of the involved devices.\\

The first and most famous of these device-independent assessments concerns Bell nonlocality~\cite{Bell64}. While being a major breakthrough in modern physics, Bell nonlocality is a weak form of device-independent certification. It merely states that the devices admit no locally causal description. Today, it is known that more precise statements can be made: the presence of entanglement~\cite{Bancal11}, randomness~\cite{Colbeck07, Pironio11} and shared secrecy~\cite{Ekert91, Acin07} can all be guaranteed in a device-independent manner. Furthermore, we now know that it is sometimes possible to get a full description of a state (up to irreducible equivalence relations) while remaining in a device independent context~\cite{tsirelson_80, popescu_92, braunstein_92, mayers_04, jed_16, coladangelo_17, Supic_17}.  This task came to be known as self-testing and does not only apply to the quantum description of states, but also to measurements~\cite{mayers_04, Yang_14, Bamps_15, Kaniewski_17} and operations, relevant for example for quantum computation~\cite{magniez_06, reichardt_13, hajduvsek_15, fitzsimons_18, coladangelo_17_bis, Sekatski_18}.\\

In this letter, we focus on the self-testing of joint measurements. In particular, we are interested in Bell state measurements which enable entanglement swapping and are thus a key ingredient in long-distance quantum communication based on quantum repeaters~\cite{Sangouard11}. To reach the point of implementing a quantum repeater, it is crucial to proceed in a scalable way and certify that each new component is qualified, independently of the purpose for which that larger device could be used. We here make a step along this line by showing how to certify the quality of Bell state measurements device-independently.\\

A few years ago, a first non-robust device-independent protocol was proposed to demonstrate that a measurement device is entangled, that is, has entangled eigenstates~\cite{Rabelo11}. Two results later showed that the entangled quality of a joint measurement can be certified in presence of some amount of noise~\cite{Bancal_15, Wan_17}. However, none of these results certify the quality of Bell state measurements since joint measurements that differ significantly from a Bell state measurement can also be entangled. Here, following the robustness definition given in~\cite{Sekatski_18}, we self-test the quality of Bell state measurements. Our certificate not only guarantees that a measurement is entangled, but also that it performs close to a perfect projection on maximally two-qubit entangled states. Our technique applies to both deterministic and probabilistic Bell state measurements. It is also noise-tolerant, hence opening a way towards the first experimental self-testing of Bell state measurements.\\

\paragraph{General description of a local measurement on a bi-partite state--} In the quantum formalism, a measurement $\mathcal{M}$ with $n$ outcomes is described by a collection of $n$ completely positive maps such that their sum is trace preserving. In our case, we are not interested in characterizing the posterior state of the measured subsystem, and a simpler description in terms of positive-operator valued measure (POVM) elements is possible. A POVM over a Hilbert space $\cH_\sA$ is a collection of positive semi-definite operators $\{ E_k \succeq 0  \}_{k=0}^{n-1} \in L(\cH_\sA)$ satisfying $\sum_k E_k = \Id$. Given the state of a system $\rho\in L(\cH_\sA\otimes \cH_\sB)$ shared between two sides $\sA$  and $\sB,$ the probability for $\sA$ to observe an outcome $k$ is given by the Born rule $p_k=\tr {E_k \otimes \Id_\sB \, \rho}$. The post measured state of $\sB$ is 
\be
\varrho_{k}=\frac{1}{p_k} \t{Tr}_\sA \left( E_k \otimes\Id_\sB \, \rho\right).
\ee
Hence, the measurement operation corresponding to outcome `k' can be described as a completely positive map
\begin{align}\label{eq: M branch}
M_k : &L(\cH_\sA\otimes\cH_\sB) \to L(\cH_\sB)\\ \nonumber
 & \rho \mapsto \t{Tr}_\sA \left(E_k\otimes\Id_\sB \, \rho\right).
\end{align}
Adding the output label `k' in a classical system with $n$ possible states, which we represent in a `n'-dimensional Hilbert space $\cH_L$, allows one to describe the whole measurement $\mathcal{M}$ in the succinct form
\begin{align}\label{eq: M full}
\mathcal{M}: \,&L(\cH_\sA\otimes\cH_\sB) \to L(\cH_\sB \otimes \cH_L)\\ \nonumber
&\rho \to \sum_k M_k(\rho) \otimes \proj{k}.
\end{align}\\

\paragraph{Description of an ideal Bell state measurement: The reference--}
A Bell state measurement (BSM) is a measurement performed jointly on two qubit subsystems. Therefore, it acts on the tensor-structured Hilbert space $\cH_\sA=\mathds{C}^2\otimes \mathds{C}^2$. Its four POVM elements are projectors onto the four Bell states:
\begin{align}\label{eq: def BSM}
&\overline{\mathcal{M}}: \rho \mapsto \sum_{k=0}^3 \text{Tr}_\sA\left( \rho \, \proj{\phi_{\bf k}}\otimes \Id \right) \otimes \proj{k}\\ 
\nonumber
&\text{with} \quad \ket{\phi_{j\ell}}=(\sigma_Z)^j\otimes (\sigma_X)^\ell\frac{1}{\sqrt{2}} \left( \ket{0, 0} + \ket{1,1} \right),
\end{align}
where $\sigma_X$ and $\sigma_Z$ are Pauli matrices. Here, we write ${\bf k}=j\ell$ with $j,\ell=0,1$ as the binary representation of $k=0,\dots, 3$, that is $\ket{\phi_{\bf k}} \equiv \ket{\phi_{j\ell}}.$\\

\begin{figure}
\begin{center}
\includegraphics[width=0.98\columnwidth]{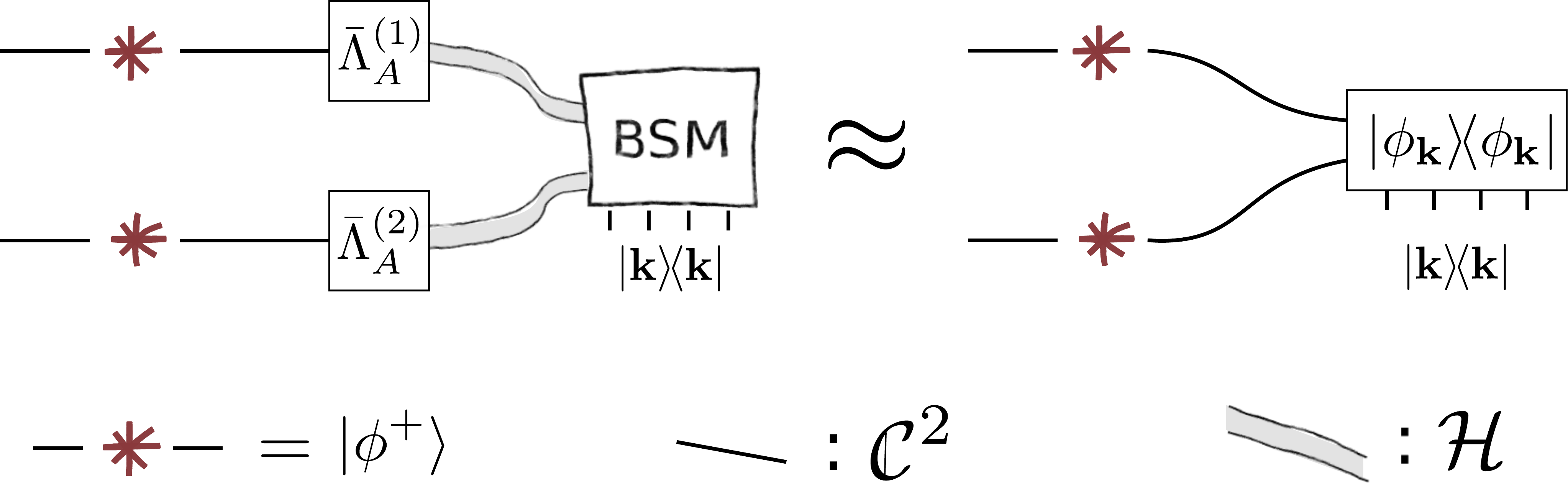}
\caption{To quantify the quality of an unknown device supposedly performing a Bell state measurement (box BSM on the left), the action of this black box supplemented with injection maps $\Lambda_A^{(1)} \otimes \Lambda_A^{(2)}$ is compared with the action of an ideal Bell state measurement (on the right) on two halves of maximally entangled two-qubit states $|\phi^+\rangle \equiv |\phi_{00}\rangle$. When the physical state admits a tensor structure across the input of the BSM, i.e. $\rho=\rho^{(1)}\otimes\rho^{(2)}$, the injection maps can incorporate Alice's marginal states, i.e $\bar \Lambda_A^{(i)}[\tau] = \Lambda_A^{(i)}[\rho_A^{(i)}\otimes \tau]$ as shown here. The red stars represent each a source producing a maximally entangled two-qubit state $\ket{\phi_{00}}$. The thin black lines correspond each to a qubit whereas the bold grey lines are associated each with a Hilbert space of unknown dimension.}
\label{fig: 1}
\end{center}
\end{figure}

\paragraph{Device-independent certification of a Bell state measurement: Formulation--} In practice, a measurement box $\mathcal{M}$ performing a BSM does not act on two qubits, but on two systems that can be well identified, e.g. two photons coming from two different optical fibres, but which live in Hilbert spaces of a larger (unknown) dimension $\cH_{A}^{(1)}$ and $\cH_{A}^{(2)}$. Hence, the total Hilbert space is $\cH_\sA=\cH_{A}^{(1)}\otimes \cH_{A}^{(2)}$. Such a measurement box $\mathcal{M}$ is certified to be a BSM if there exist local CPTP injection maps 
\be
\Lambda_A^{(i)}:  L(\cH_A^{(i)}\otimes\mathds{C}^2)\to L(\cH_A^{(i)}), \-\ i=1,2
\ee
such that 
\begin{align}\label{eq: M ideal}
\mathcal{M}\circ(\Lambda_A^{(1)}\otimes \Lambda_A^{(2)})[\rho_\sA \otimes \tau]
=\overline{\mathcal{M}}[\tau]
\end{align}
for all state $\tau$ on the extended Hilbert space $\dC^{2}\otimes\dC^{2}\otimes \cH_\sB$. Here, $\rho_\sA = \text{Tr}_{\sB}(\rho)$ is the partial state of $\sA$ on which $\mathcal{M}$ acts, $\rho$ being the state of the source used in practice, 
and $\overline{\mathcal{M}}$ the reference measurement defined in Eq.~\eqref{eq: def BSM}. For simplicity, the identity operator on side $\sB$ is not written in Eq.~\eqref{eq: M ideal}. The certificate in Eq.~\eqref{eq: M ideal} shows that the device $\mathcal{M}$ performs a perfect Bell state measurement on two qubits. In order to be used in practice on two qubits, the qubits first have to undergo appropriated local injection maps.\\ 

No real-world device operates ideally. Hence, a sensible notion of certification has to cover non-ideal cases. Following Ref.~\cite{Sekatski_18}, we consider the Choi fidelity $\cF(\mathcal{M},\overline{\mathcal{M}}),$ that is, the Uhlmann fidelity between the states obtained by acting on half of a maximally entangled state with either the actual measurement $\mathcal{M}$ combined with the injection maps $\Lambda_A^{(i)}$ or simply the ideal measurement $\overline{\mathcal{M}}$ (see Fig.~\ref{fig: 1}). This quantifies the deviation of $\mathcal{M}$ from the ideal case. In the case of a complete and deterministic Bell state measurement, i.e. a measurement distinguishing the four Bell states with unit efficiency, this fidelity takes the form
\begin{align}\label{eq:F BSM}
\cF(\mathcal{M},\overline{\mathcal{M}}) = \underset{\Lambda_A^{(1)}, \Lambda_A^{(2)}}{\max}\ F\Big(&\mathcal{M}\circ(\Lambda_A^{(1)}\otimes \Lambda_A^{(2)})\left[\rho_\sA \otimes \ket{\phi_{00}}^{\otimes 2} \right], \nonumber\\
&\frac{1}{4} \sum_{k} \proj{\phi_{\bf k}} \otimes \proj{k}\Big)
\end{align}
where $F(\rho,\sigma)=\tr{\!\sqrt{\!\!\sqrt{\rho}\,\sigma\!\sqrt{\rho}}}$ is the Uhlmann fidelity between two states $\rho$ and $\sigma$.\\

\paragraph{Device-independent certification of a complete and deterministic Bell state measurement: Recipe--} In this section, we show how the quantity given in Eq.~\eqref{eq:F BSM} can be bounded experimentally. The basic recipe uses two steps as depicted in Fig.~\ref{fig: 2}. A source is used to prepare a four-partite state, ideally two pairs of two-qubit maximally entangled states. In the first step, two of these subsystems are sent to party $\sA$ while the two remaining ones are distributed to additional parties $B^{(1)}$ and $B^{(2)}$ who measure them locally. The aim of the first step is to characterize the post-measured state shared between $B^{(1)}$ and $B^{(2)}.$ In the second step, the same source is characterized independently via local measurements performed by the four parties $A^{(1)}$, $A^{(2)}$, $B^{(1)}$ and $B^{(2)}$.\\

Before detailing the steps I and II, we specify the notations. We denote by $\rho \in L(\cH_A\otimes \cH_B)$ the state produced by the source, where $\cH_\sA=\cH_A^{(1)}\otimes \cH_A^{(2)}$ and $\cH_\sB=\cH_B^{(1)}\otimes \cH_B^{(2)}$ are the Hilbert spaces on each side, and the tested joint measurement $\mathcal{M}$ is of the form \eqref{eq: M full}. The post-measured state of parties $B^{(1)}$ and $B^{(2)}$ on side $\sB$ conditioned on outcome $k$ is given by 
\be
\varrho_k = \frac{ M_k[ \rho ]}{\t{Tr} \,M_k[ \rho ]} \in L(\cH_B^{(1)}\otimes \cH_B^{(2)}),
\ee
and occurs with a probability $p_k =\t{Tr} \,M_k[ \rho ]$.\\

Let us now focus on step I. In the ideal case, where maximally entangled two qubit states are produced by the source and the joint measurement is a Bell state measurement, the state of $B^{(1)}$ and  $B^{(2)}$ is projected onto one of the four Bell states depending on the outcome $k.$ For each $k$, there exist measurements for $B^{(1)}$ and $B^{(2)}$ with two possible settings $y_{1(2)}=0$ or $1$ and two possible outputs  $b_{1(2)}=0$ or $1$ such that the Clauser, Horne, Shimony, and Holt (CHSH) value~\cite{CHSH69}
\be
\beta_k=\sum_{b_1,b_2,y_1,y_2} (-1)^{b_1+b_2+y_1\cdot y_2} p(b_1b_2| y_1 y_2k)
\ee
attains its maximum quantum value $2\sqrt{2}$ with a simple relabeling of settings and outputs. For example, $2\sqrt{2}$ is obtained with the state $\ket{\phi_{00}}$ when $B^{(1)}$ and $B^{(2)}$ use the measurements $\{\sigma_Z, \sigma_X \}$ and $\{\frac{\sigma_X+\sigma_Z}{\sqrt{2}},\frac{\sigma_X-\sigma_Z}{\sqrt{2}} \}$ respectively. It is straightforward to verify that the exchange of the outputs of the first measurement $\sigma_X$ of $B^{(1)}$ and the exchange of the roles of $B^{(2)}$'s measurements correspond to local unitary transformations, respectively $\sigma_Z$ and $\sigma_X$, applied on the respective sides. Hence, the same measurements can be used, with some relabelling of settings and outputs, to violate maximally the CHSH inequality with the states $\Id\otimes\sigma_X \ket{\phi_{00}}\equiv \ket{\phi_{01}},$ $\sigma_Z\otimes\Id \ket{\phi_{00}}\equiv \ket{\phi_{10}}$ and $\sigma_Z\otimes\sigma_X \ket{\phi_{00}}\equiv \ket{\phi_{11}}$, which are precisely the three Bell states heralded by outcome $k=1,2$ and $3$ respectively.\\

Using the results of Ref.~\cite{jed_16}, we show in the Appendix A that given the four values $\beta_k$ obtained for the four versions of CHSH, there exist local extraction maps $\Lambda_B^{(1,2)}:L(\cH_B^{(1,2)})\to \dC^2$  such that 
\begin{equation}\label{eq: F CHSH}
F\left( (\Lambda_B^{(1)} \otimes \Lambda_B^{(2)})[\varrho_k], \ket{\phi_{\bf k}}\right) \geq F^o_{k} = \sqrt{1 -\frac{1}{2}\cdot\frac{2\sqrt{2}-\beta_k}{2\sqrt 2 -\beta^*}},
\end{equation}
where $\beta^*=\frac{2(8 + 7\sqrt{2})}{17} \approx 2.11$. These isometries can be explicitly defined from the quantum description of the measurements performed by $B^{(1)}$ and $B^{(2)}$ respectively, and importantly they are the same for the four Bell states.\\

\begin{figure}
\begin{center}
\includegraphics[width=0.99\columnwidth]{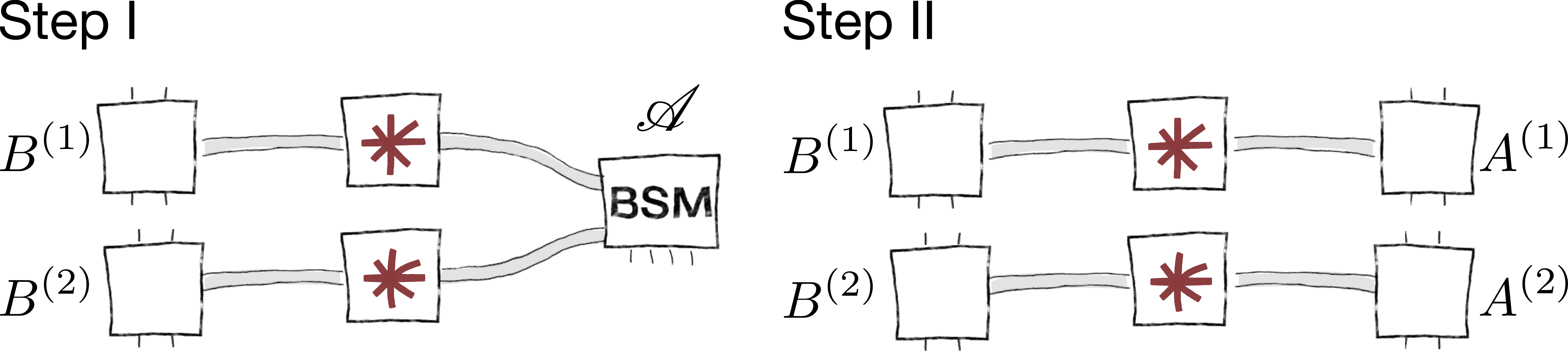}
\caption{Recipe for the device-independent certification of a Bell state measurement with two independent sources. Two independent sources prepare ideally two-qubit entangled states. In the first stage (step I) half of these states are sent to $\sA$ performing the BSM to be certified. The two remaining systems are sent to $B^{(1)}$ and $B^{(2)}$, performing a Bell test in order to certify their state conditional on the outcome of the BSM. In the optional second stage (step II), the states produced by the sources are characterized independently via a four-partite Bell test. The parties on side $\sB$ use the same measurements for both stages of the experiment. Note that when certifying a deterministic BSM with both step I \& II, the independence of the sources is not required.}
\label{fig: 2}
\end{center}
\end{figure}

Just like in the definition of $\mathcal{M}$, the output of Alice can be included as a label in a global post-measurement state
\be
\varrho = \sum_k p_k\,\varrho_k\otimes \proj{k},
\ee
where $p_k$ are the probabilities of each outcome observed in the experiment. Using the same isometries $\Lambda_B^{(1,2)}$, the certificates for the four branches $\varrho_k$ of Eq.~\eqref{eq: F CHSH} can thus be combined into a single certificate for the global output state
\begin{align}
\nonumber
& F\left((\Lambda_B^{(1)} \otimes \Lambda_B^{(2)}\otimes \Id_{\cH_L})[\varrho] , \sum_k\frac{1}{4} \proj{\phi_{\bf k}}\otimes \proj{k}\right) \\ 
\label{eq: Fo}
& \geq F^o = \sum_k \sqrt{\frac{p_k}{4}} \,F^o_k.
\end{align}
Note that we here used the orthogonality of states $\ket{k}.$ This relation is not yet enough to certify the measurement box $\mathcal{M}$ itself. In particular, it does not give the form of maps $\Lambda_A^{(i)}$ that are required in Eq.~\eqref{eq:F BSM}, see Fig.~\ref{fig: 1}. The most straightforward way to do so is to add in a careful analysis $F^o$ with the data collected in step II, as we detail now.\\

Step II refers to a characterization of the state created by the source. Using the Bell inequality\footnote{$B_\varphi$ is a function of expectation values of outcomes given the inputs such that Eq. (14) is a true Bell inequality. We invite the reader to look at Ref.~\cite{Sekatski_18} for the details of this Bell inequality suited e.g. for certifying a tensor product of two singlets when $\varphi=0.$} proposed in Ref.~\cite{Sekatski_18} 
\begin{equation}
B_\varphi \leq \frac{(\sqrt{2}+1)(\cos\varphi+\sin\varphi)+2}{5\sqrt{2}}
\end{equation}
with $\varphi=0$, the fidelity of the input state $\rho$ with the maximally entangled state $\ket{\phi_{00}}^{\otimes 2}$ can be bounded as
\begin{equation}\label{eq: Fi}
F((\widetilde{\Lambda}_A \otimes \Lambda_B) [\rho],|\phi_{00}\rangle^{\otimes 2}) \geq F^i = \sqrt{\frac{1}{4}\left(1 + 3\frac{\delta-\delta^*}{1-\delta^*}\right)}.
\end{equation}
$\delta=\langle B_{\varphi=0}\rangle$ is the Bell inequality violation, $\delta^* \geq 0.744$, and the isometries have the product structure $\widetilde{\Lambda}_A = \widetilde{\Lambda}_A^{(1)} \otimes \widetilde{\Lambda}_A^{(2)}$, $\Lambda_B = \Lambda_B^{(1)} \otimes \Lambda_B^{(2)}.$ Importantly, this Bell test can be performed by keeping the same measurements at side $\sB$ as used in step I, such that the certificates for the sources come with the same extraction maps $\Lambda_B^{(1)}$ and $\Lambda_B^{(2)}$ as those certifying the post-measured state in Eq.~\eqref{eq: Fo}.
Under this condition, the two certificates for the states before and after the measurement $\mathcal{M}$ can be combined in a certificate of the BSM itself~\cite{Sekatski_18}, leading to 
\be
\label{fid_BSM}
\cF(\mathcal{M}, \overline{\mathcal{M}})\geq \cos\left(\arccos (F^o) + \arccos(F^i)\right).
\ee 
When $\beta_k=2\sqrt{2},$ that is, $F^o_k=1$ for the four values of $k,$ $p_k=1/4$ and $\delta=1$ is maximal, this bound guarantees that $\cF(\mathcal{M},\overline{\mathcal{M}})=1.$ In noisy scenarios, 
the fidelity that can be certified is shown in Fig. \ref{fig: 3}. The full line shows $\cF(\mathcal{M}, \overline{\mathcal{M}})$ as a function of $\beta_k=\beta\ \forall k$ in the case where the source delivers two-qubit maximally entangled states $(\delta=1)$ while the dashed line corresponds to $\delta=\beta_k/2\sqrt{2}$, i.e. the case where all measurements are perfect but the source is noisy. This shows the robustness of the bound given in Eq.~\eqref{fid_BSM} when applied to noisy Bell state measurements. The resistance to noise operating at the level of the source is not as strong since it impacts both $F^i$ and $F^o.$\\

\begin{figure}
\begin{center}
\includegraphics[width=1\columnwidth]{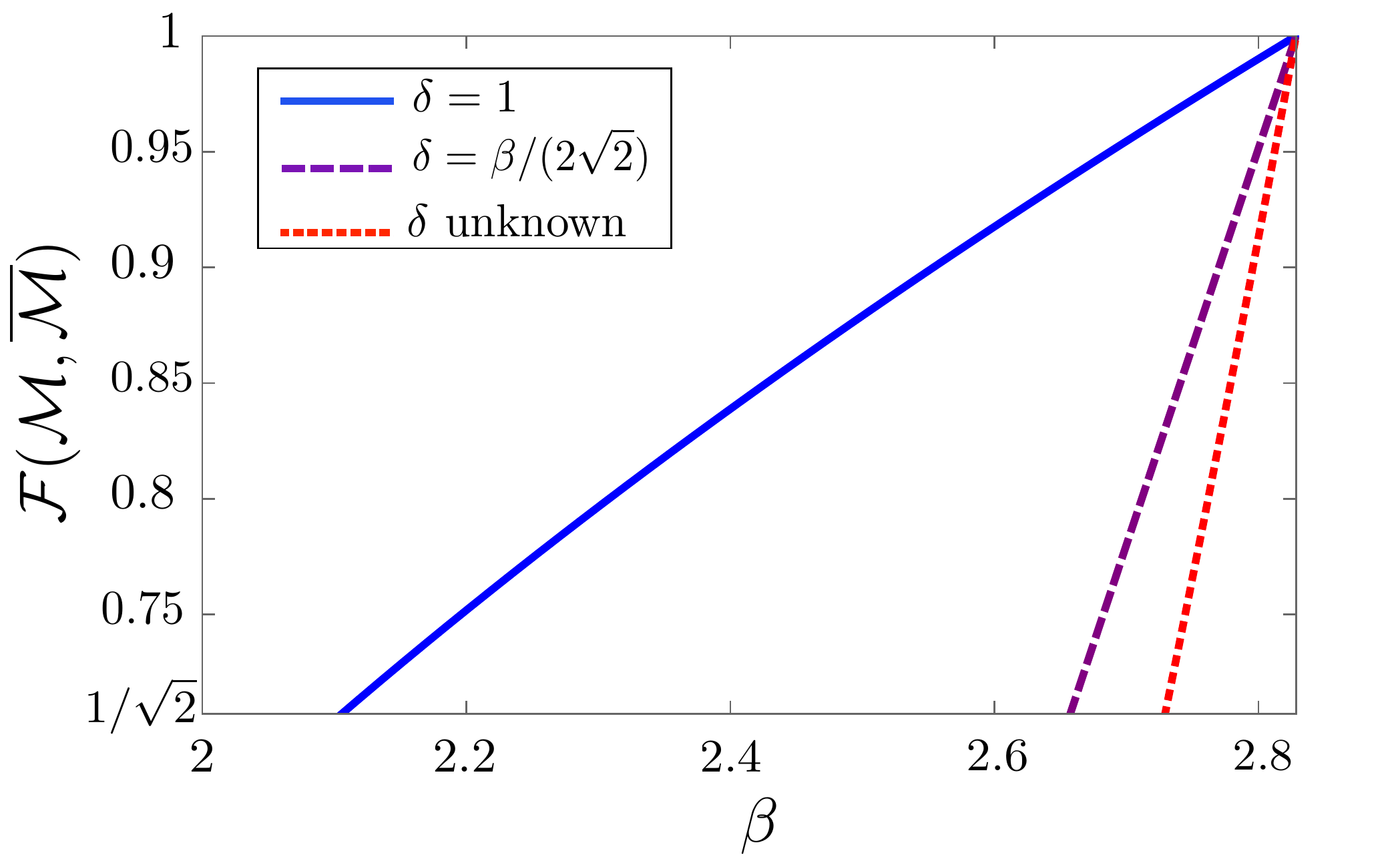}
\caption{Fidelity of a deterministic and complete Bell state measurement as a function of the CHSH value $\beta_k=\beta$ obtained between parties $B^{(1)}$ and $B^{(2)}$ when performing measurements on the post-measurement state (Step I). Here, we assume that all outcomes are equiprobable, i.e. $p_k=1/4$. The full (blue) line is obtained when the value of the 4-partite Bell expression observed in Step II is maximum $(\delta=1)$. The dashed (purple) line is obtained assuming $\delta=\beta/2\sqrt{2}.$ The dashed (red) line with smaller dashes relies on post-measurement statistics only, but requires that the measured state $\rho$ admit a tensor structure.}
\label{fig: 3}
\end{center}
\end{figure}

Note that it is possible to bound the fidelity of the Bell state measurement by taking into account only the CHSH values on the post measurement state, that is, without using information on $F^i.$ However, this requires the initial state to come from two independent sources, i.e. {\bf $\rho=\rho^1\otimes \rho^2$}. Indeed, if the source can produce the state $\rho = \sum_k {\ketbra{k}{k}}_{\sA} \otimes {\ketbra{\phi_{\bf k}}{\phi_{\bf k}}}_{\sB}$, all post-measurement statistics can be reproduced without any BSM. Under this assumption, we show in the Appendix B that 
\begin{align}
&\cF(\mathcal{M},\overline{\mathcal{M}})\geq
\nonumber
\\
& \, \cos\left[\arccos(F^o) +\arccos\left(\left( \sum_k p_k \left(F^o_k\right)^2\right)^2\right)\right].
\nonumber
\end{align}
Note that this bound is quite demanding. When $p_k=1/4$ and $\beta_k=\beta$, it requires $\beta$ higher than $\approx 2.73$ to certify that $\cF(\mathcal{M},\overline{\mathcal{M}}) > 1/\sqrt{2}$, see Fig.~\ref{fig: 3}.\\

\paragraph{Device-independent certification of a probabilistic/partial Bell state measurement--} In many situations, e.g. for measurement setups based on linear optics, the realization of a full Bell state measurement is neither reachable nor necessary. Instead, one is happy enough with a partial BSM where one outcome or a subset of the outcomes heralds the successful projection onto some of the four Bell states. Moreover, the projection need not always be successful (probabilitic BSM). In this case, the POVM corresponding to the successful outcomes is then of the form
\be \label{eq: aBSM povm}
\overline{E}_k=\zeta_k \proj{\phi_{\bf k}},
\ee
where the proportionality factors $\zeta_k$ is equal to the probability to obtain the outcome $k$ when the measured state is $\ket{\phi_{\bf k}}$. In practice, this probability can be small but what matters most is the quality of the corresponding projector.\\

Consider the outcome $k=0$ and associated completely positive map $M_0$. 
Analogously to the case before, with the help of local isometries the quality $\cF_\text{cond}(M_0, \overline{M_0})$ of the measurement $M_0$ with respect to the ideal projection described by $\overline{M}_0: \rho \mapsto \t{tr}_{\sA}\left(\rho \overline{E}_0 \otimes \Id_\sB \right)$ can be assessed conditioned on the successful behavior of $M_0$. A formal definition of $\cF_\text{cond}$ can be found in Appendix C, where we also show how it can be bounded from the Bell values $\beta_0$ and $\delta$  obtained in step I and step II respectively. For the relevant regime $(F^i)^2 + p_0 \geq 1,$ we obtain
\begin{align}\label{eq: F conditional}
\mathcal{F}_\text{cond} &\left(M_0, \overline{ M_0} \,\right)
\geq\\ &\cos\left(\arccos{(F^o_0)}
\nonumber
+\arccos \left(\sqrt{\frac{p_0 + (F^i)^2-1}{p_0}}\right)\right)
\nonumber
\end{align}
and
\begin{eqnarray}\label{eq: zeta0}
\zeta_0 \geq 4 \left(\sqrt{p_0 \left(F^i\right)^2} - \sqrt{(1-p_0)\left(1-\left(F^i\right)^2\right)}\right)^2,
\end{eqnarray}
where $F^o_0,$ $F^i$ are bounded by $\beta_0$ and $\delta$ via Eq.~\eqref{eq: F CHSH} and \eqref{eq: Fi} respectively. $p_0=\tr{M_0 [\rho]}$ is the experimentally observed  probability of obtaining the measurement outcome `$0$', and $\zeta_0/4$ is the success probability of the measurement $M_0$ preceded by the injection maps when applied on $\ket{\phi_{00}}^{\otimes 2}$. Detailed comments on these results are given in Appendix C.\\

\paragraph{Conclusion--} We have discussed the device independent certification of Bell state measurements. After proposing a formal definition of what such a certificate means for complete and deterministic Bell state measurements, we have provided a concrete recipe based on Bell tests leading to noise-tolerant certificates. The formulation and recipe has then been extended to partial and probabilistic Bell state measurements. These results could play an important role in the implementation of quantum networks by opening the way of a scalable approach consisting in certifying each building block in a device-independent way.\\

\begin{acknowledgments}
This work was supported by the Swiss National Science Foundation (SNSF), through the NCCR QSIT, Grant PP00P2-150579, PP00P2-179109 and 200021-175527. We also acknowledge the Army Research Laboratory Center for Distributed Quantum Information via the project SciNet.
\end{acknowledgments}

\paragraph{Note added--} While finishing this manuscript, we became aware of related work by Marc Olivier Renou et al. [ref].\\

\hrule
\ \\

\paragraph{Appendix A: Self-testing the four Bell states with the same measurements--} The aim of this appendix is to provide a detailed proof of the formula (10) of the main text. We start with a quick reminder of the results presented in Ref.~\cite{jed_16}. We then show that the four Bell states $\varrho_k$ can be self-tested with the same measurement boxes and local maps. We conclude with a proof for Eq. (10) of the main text.\\

Consider the qubit measurement operators $A_0(a), A_1(a), B_0(b)$ and $B_1(b)$ defined as
\be\label{eq:Jordan}
A_r(a) = \cos(a) \sigma_X+ (-1)^r \sin(a) \sigma_Z
\ee
for $a \in [0,\pi/2]$ (and the same for $B$), and the extraction map defined as
\begin{align}\label{eq:JED channel}
\Lambda_\lambda [\rho] &= \frac{1+g(\lambda)}{2} \rho +\frac{1-g(\lambda)}{2} \sigma_\lambda \rho\, \sigma_\lambda\\
&\text{with}\quad \sigma_ \lambda= 
\begin{cases} \sigma_X & \lambda\leq \pi/4\\
\sigma_Z & \lambda> \pi/4,
\end{cases}
\\  &g(\lambda)= (1+\sqrt{2})(\sin \lambda +\cos \lambda +1).
\end{align} 
Further consider the operator $W$ that yields the CHSH value $\beta$ obtained with the measurements defined above
\be
W_{a,b} = \sum_{r,t =0,1} (-1)^{rt} A_r(a) \otimes B_t(b)
\ee
It has been shown in~\cite{jed_16} that for $s=(4+5\sqrt{2})/16$, $\mu = -(1+2\sqrt{2})/4$ and the maximally entangled state
\be\label{eq:JED state}
\ket{\Psi}= \underbrace{- e^{\ii \frac{\pi}{8} \sigma_Y} \sigma_X}_U
\otimes \Id\ket{\phi_{00}},
\ee
the inequality
\be\label{eq:JED op}
(\Lambda_a\otimes \Lambda_b) \left[\proj{\Psi}\right] -  s \, W_{a,b} -\mu\, \Id\geq 0
\ee
holds for all values of $a$ and $b.$ This operator inequality is directly useful for self-testing. It implies in particular that for any two-qubit state $\rho$  
\begin{align}
\nonumber
\tr{\rho \, (\Lambda_a\otimes \Lambda_b) \left[\proj{\Psi}\right]} &\geq s \,\tr{\rho W_{a, b}} +\mu \\
\implies F^2\big((\Lambda_a \otimes \Lambda_b)[\rho],\ket{\Psi} \big) &\geq s \beta +\mu, \label{eq: F JED}
\end{align}
where we used the fact that $\Lambda_{\lambda}$ is self-dual. Moreover, using the Jordan lemma, the self-testing of states of arbitrary dimension can be reduced to the qubit case \cite{scarani_12}, and it can be shown that Ineq.~\eqref{eq: F JED} holds in general, without assumption on the Hilbert space assumption \cite{jed_16}.\\

We now prove that the four Bell states give the maximum quantum value of CHSH with the same measurement boxes upon relabeling the outputs of $A_0\to -A_0$ and/or exchanging the inputs of Bob, i.e. permuting $B_0$ and $B_1$. Let us call these two transformations of the CHSH observable $T_A$ and $T_B.$ We emphasize that they simply correspond to different post-processing of observed statistics. Using Eq.~\eqref{eq:Jordan}, we can compute how the Bell operator changes with the different post-processings:
\begin{align}
\nonumber
T_B(W_{a,b}) &= \Id\otimes U_B[W_{a,b}]= (\Id\otimes \underbrace{\sigma_X}_{=U_B}) \,W_{a,b}\, (\Id\otimes \sigma_X)\\
\nonumber
T_A(W_{a,b}) &= U_A\otimes\Id[W_{\frac{\pi}{2}-a,b}] \nonumber\\ 
\nonumber
&=(\underbrace{e^{-\ii \frac{\pi}{4}\sigma_Y }\sigma_X}_{=U_A} \otimes \Id) \,W_{\frac{\pi}{2}-a,b}\, (\sigma_X e^{\ii \frac{\pi}{4}\sigma_Y } \otimes \Id).
\end{align}
Moreover, we can show from Eq.~\eqref{eq:JED channel} that
\begin{align}
U_B\circ\Lambda_b &= \Lambda_b\circ U_B\\
U_A\circ\Lambda_a &= \Lambda_{\frac{\pi}{2}-a}\circ U_A.
\end{align}
Hence applying the unitary transformation $U_A$, $U_B$ and $U_A\otimes U_B$ to the inequality \eqref{eq:JED op} for $a$ and $a'=\frac{\pi}{2}-a$ yields
\begin{align}
\nonumber
(\Lambda_a\otimes \Lambda_b) &\left[\proj{\Psi_{01}}\right] - s \, T_B(W_{a,b}) -\mu\, \Id\geq 0\\
\nonumber
(\Lambda_{a}\otimes \Lambda_b)& \left[\proj{\Psi_{10}}\right] - s \, T_A(W_{a,b}) -\mu\, \Id\geq 0\\
\nonumber
(\Lambda_{a}\otimes \Lambda_b) &\left[\proj{\Psi_{11}}\right] - s \, T_A\otimes T_B(W_{a,b}) -\mu\, \Id\geq 0
\end{align}
with $\ket{\Psi_{j\ell}}= (U_A)^{j}\otimes (U_B)^\ell \ket{\Psi}$. Finally, using Eq.~\eqref{eq:JED state} and the equality $U^\dag U_A U =\sigma_Z,$ we obtain the desired result
\be
\ket{\Psi_{j\ell}} = (U\otimes \Id) (\sigma_Z^j\otimes \sigma_X^\ell) \ket{\phi_{00}}= (U\otimes \Id)   \ket{\phi_{j\ell}}.
\ee
Hence, all the four Bell states can be certified device-independently with the same measurement boxes and the same extraction maps. The fidelity of the extracted state with the corresponding Bell state is lower bounded by the same expression
\be
F^o_k = \sqrt{1 -\frac{1}{2}\cdot\frac{2\sqrt{2}-\beta_k}{2\sqrt 2 -\beta^*}},
\ee
where $\beta_k$ are values of different CHSH tests that are related by mere relabelings $T_A$ and $T_B$ of some inputs and some outputs of the measurement boxes and $\beta^*=\frac{2(8 + 7\sqrt{2})}{17} \approx 2.11$. This means that a non trivial bound $(F^o_k> 1/\sqrt{2})$ is obtained as soon as $\beta_k>\beta^*.$\\

\paragraph{Appendix B: Certification of a complete and deterministic Bell state measurement from only post-measurement statistics--} The goal of this appendix is to show how to certify a Bell state measurement when the quality of sources cannot be estimated. In particular, we show that the step I described in the main text is sufficient to establish a lower bound on the quality of a complete and deterministic Bell state measurement. As discussed in the main text, this requires assuming that two independent sources are used with respect to the inputs 1 and 2 of the Bell state measurement.\\

Let $\rho_i$ be the state of the source used as input $i=1,2$ of the Bell state measurement, that is $\rho=\rho_1 \otimes \rho_2.$ We further define  $\rho_i' \in L(\cH_A^{(i)}\otimes\dC^2)$ as the state produced by the source $i$ after the extraction maps introduced in Eq.~\eqref{eq: F CHSH} is applied on side $\sB$, i.e.
\be
\rho'_i = \Id\otimes\Lambda_B^{(i)}[\rho_i].
\ee
This state admits a purification $\ket{\Psi_i}\in \cH_A^{(i)}\otimes\cH_\t{E}^{(i)}\otimes\dC^2$ on an extended Hilbert space, i.e. $\rho_i'=\t{Tr}_\t{E}\proj{\Psi_i}$. As $B^{(i)}$ beholds a qubit in the $B^{(i)}| A^{(i)}\t{E}^{(i)}$ splitting, the purified states admit a Schmidt decomposition with two components only:
\be\label{eq:source pur}
\ket{\Psi_i} = \sqrt{q_i} \ket{0}_{B^{(i)}}\ket{0}_{A^{(i)}\t{E}^{(i)}}+ \sqrt{1-q_i}\ket{1}_{B^{(i)}}  \ket{1}_{A^{(i)}\t{E}^{(i)}},
\ee
where $\ket{0}_{B^{(i)}}$ and $\ket{1}_{B^{(i)}}$ define a basis in  $\dC^2.$ $\ket{0}_{A^{(i)}\t{E}^{(i)}}$ and $\ket{1}_{A^{(i)}\t{E}^{(i)}}$ are two orthogonal states in $\cH_A^{(i)}\otimes\cH_\t{E}^{(i)}$.\\

Similarly, we now consider the qubit state corresponding to the conditional state $\varrho_k$ after the extraction maps on side $\cB$
\begin{equation}
    \varrho_k'= \Lambda_B^{(1)} \otimes \Lambda_B^{(2)} [\varrho_k].
\end{equation}
We introduce the unitaries $U_k = U_k^{(1)} \otimes U_k^{(2)}$ such that $U_k|\phi_{\bf k}\rangle = |\phi_{00}\rangle.$ The state 
\be
\sigma=\sum_k p_k U_k \varrho_k' \, U_k^\dag
\ee
then satisfies
\be\label{eq: F Neg bound}
\bra{\phi_{00}} \sigma \ket{\phi_{00}} =\sum_k p_k \bra{\phi_{\bf k}}\varrho_k'
 \ket{\phi_{\bf k}} \geq \sum_k p_k \left(F^o_{k}\right)^2,
\ee
and can be prepared from $\ket{\Psi_1} \otimes \ket{\Psi_2}$ by local operations and classical communications (LOCC) in the splittings $B^{(1)}|\sA\t{E}^{(1)}\t{E}^{(2)}|B^{(2)}.$ Therefore, $\sigma$ can also be prepared from $\Psi_1$ ($\Psi_2$) by LOCC in the splitting $B^{(1)}\sA\t{E}^{(1)}\t{E}^{(2)}|B^{(2)}$  ($B^{(1)}|\sA\t{E}^{(1)}\t{E}^{(2)}B^{(2)}$). Hence, for any entanglement measure $E_M$  the inequality
\be
\label{Ent_measure}
E_M(\ket{\Psi_i})\geq E_M( \sigma)
\ee
holds for $i=1,2$.\\

Let us consider the Negativity $N$. It has the property that for any two-qubit state $\rho$~\cite{Verstraete02}
\begin{equation}
    N(\rho) = \frac{||\rho^{T_B}||_1 - 1}{2} \geq \bra{\phi_{00}}\rho \ket{\phi_{00}} -\frac{1}{2}.
\end{equation}
Using $N(\Psi_i)=\sqrt{q_i(1-q_i)}$, Eq~\eqref{eq: F Neg bound}, Eq.~\eqref{Ent_measure}, and the fact that the negativity is an entanglement measure, we end up with 
\be
\label{bound_fid}
\sqrt{q_i(1-q_i)} \geq \sum_k p_k \left(F^o_k\right)^2 -\frac{1}{2}
\ee
for both $i=1$ and $2.$\\


\begin{figure}
\begin{center}
\includegraphics[width=0.98\columnwidth]{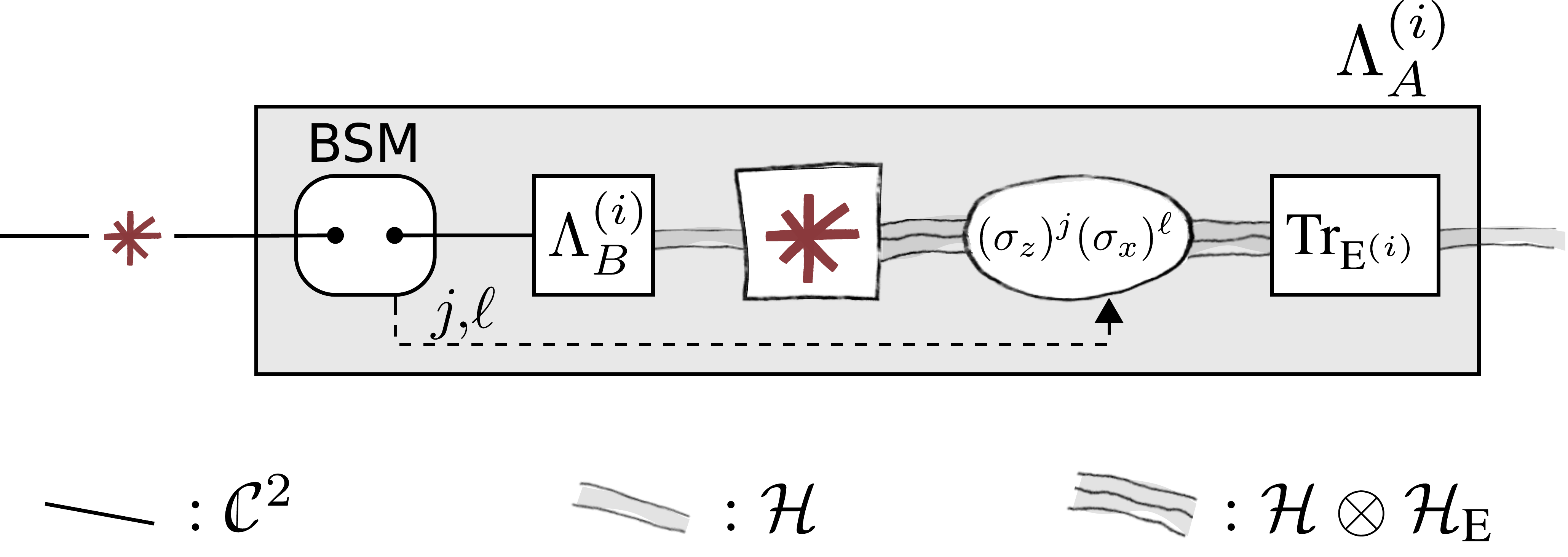}
\caption{The proposed injection maps $\Lambda_A^{(i)}$ is obtained by using the purified state of the source supplemented with the local extraction map $\Lambda_B^{(i)}$ to teleport the input qubit (half of $\ket{\phi_{00}}$ state here). Here, $\mathcal{H}_\text{E}$ is the purifying Hilbert space, which is traced out after teleportation.}
\label{figSM: 1}
\end{center}
\end{figure}

We now construct a map $\Lambda_A^{(i)}$ such that 
\be
F(\Lambda^{(i)}_A\otimes\Id[\ket{\phi_{00}}], \rho_i') =  \frac{1}{2}+\sqrt{q_ i(1-q_ i)}.
\ee
Consider the map obtained from the teleportation operation represented in Fig. \ref{figSM: 1}: a Bell state measurement is applied on the input qubit and on one part of the state of the source supplemented with the map $\Lambda_B^{(i)}$, i.e. $\ket{\Psi_i}$. The corresponding projection onto $\ket{\phi_{j\ell}}$ is identified by the classical outputs $j\ell=00, 01, 10$ and $11$. In case where $j=1$ ($\ell=1$), a $\sigma_z$ ($\sigma_x$) operation is applied in the qubit subspace having $\{|0\rangle_{A\t{E}^{(i)}}, |1\rangle_{A\t{E}^{(i)}}\}$ as canonical basis. Finally the auxiliary system $\t{E}^{(i)}$ is traced out. If the input state is initially entangled with respect to the $\ket{\phi_{00}}$ state, we find 
$$
\Lambda_A^{(i)} \otimes \Id [|\phi_{00}\rangle] = \frac{1}{2} \left(\rho_i' + \left(\sigma_x \otimes \sigma_x\right) \rho_i' \left(\sigma_x \otimes \sigma_x\right)\right)
$$
that is 
\be
\label{fid_int}
F(\Lambda^{(i)}_A\otimes\Id[\ket{\phi_{00}}], \rho_i') =  \frac{1}{2}+\sqrt{q_ i(1-q_ i)}.
\ee\\

Combining the previous result with Eq. \eqref{bound_fid},  we get 
\begin{equation}
F( \Lambda_B^{(i)}[\rho_i],\Lambda_A^{(i)} [\ket{\phi_{00}}]) \geq \sum_k p_k \left(F^o_k\right)^2
\end{equation}
and thus
\begin{align}
 F( &\Lambda_B^{(1)}\otimes \Lambda_B^{(2)}[\rho_1\otimes \rho_2],\Lambda_A^{(1)}\otimes \Lambda_A^{(2)} [\ket{\phi_{00}}^{\otimes 2}]) \nonumber\\&\geq \left(\sum_k p_k \left(F^o_k\right)^2\right)^2.
\end{align}
The processing inequality for fidelity ensures that the bound holds when the measurement $\mathcal{M}$ is applied on the Alice side of both states,
\begin{align}
 F( &\Lambda_B^{(1)}\otimes \Lambda_B^{(2)}\otimes \mathcal{M} [\rho_1\otimes \rho_2],\mathcal{M}\circ(\Lambda_A^{(1)}\otimes \Lambda_A^{(2)}) [\ket{\phi_{00}}^{\otimes 2}]) \nonumber\\&\geq \left(\sum_k p_k \left(F^o_k\right)^2\right)^2.
\end{align}
Finally, combining this bound with the output fidelity 
\be
F( \Lambda_B^{(1)}\otimes \Lambda_B^{(2)}\otimes \mathcal{M} [\rho_1\otimes \rho_2],\overline{\mathcal{M}} [\ket{\phi_{00}}^{\otimes 2}]) \geq F^o,
\ee
and using the equivalent of the triangle inequality for fidelities, we get the desired bound on the quality of the Bell state measurement:
\begin{eqnarray}
\nonumber
&\arccos(\cF(\mathcal{M},\overline{\mathcal{M}})) &\leq \arccos(F^o)
\\
&&
+\arccos\left(\left( \sum_k p_k \left(F^o_k\right)^2\right)^2\right).
\nonumber
\end{eqnarray}
\\

\paragraph{Appendix C: Certification of probabilistic and partial Bell state measurements--} The aim of this appendix is to demonstrate the relations \eqref{eq: F conditional} and \eqref{eq: zeta0} given in the main text that bound $\mathcal{F}_\text{cond} \left(M_0, \overline{ M_0} \,\right)$ and $\zeta_0$. These bounds rely on the CHSH value $\beta_0$ obtained on step I by parties $B^{(1)}$ and $B^{(2)}$ when measuring the state conditioned on the outcome $k=0$ of the BSM, on the probability $p_0$ of observing this outcome $k=0$, and on the Bell value $\delta$ obtained by all the four parties in step II. We start by giving a clear definition of self-testing for a probabilistic Bell state measurement.\\

Consider the outcome $k=0$ and associated completely positive map $M_0.$ We say that $M_0$ corresponds to a branch of a probabilistic or partial Bell state measurement if there exist local maps $\Lambda_A^{(1)}$ and $\Lambda_A^{(2)}$ and a finite $\zeta_0>0$ such that
\be \label{eq: aBSM ideal}
M_0\circ(\Lambda_A^{(1)}\otimes \Lambda_A^{(2)})[\rho_\sA\otimes\tau] = \zeta_0\ \t{Tr}_{\sA}\left( \proj{\phi_{00}}\otimes \Id \, \tau \right)
\ee
for all states $\tau$.\\ 
	
As before, we use the Choi  fidelity to quantify the quality of a probabilistic or partial Bell state measurement in any practical situation where the desired operation is not realized exactly. This allows us to define the conditional fidelity for a probabilistic channel as
\begin{align}\label{eq:F aBSM}
&\mathcal{F}_\text{cond} \left(M_0, \overline{ M_0} \,\right) = \underset{\Lambda_A^{(1)}, \Lambda_A^{(2)}}{\max}\\
&\ \frac{4}{\zeta_0}\, F\left( M_0 \circ(\Lambda_A^{(1)}\otimes \Lambda_A^{(2)})[\rho_\sA\otimes\ket{\phi_{00}}^{\otimes 2}], \overline{M}_0[\ket{\phi_{00}}^{\otimes 2}]\right) \nonumber
\end{align}
where $\zeta_0/4=\t{Tr} \big(M_0 \circ(\Lambda_A^{(1)}\otimes \Lambda_A^{(2)})[\rho_\sA\otimes\ket{\phi_{00}}^{\otimes 2}]\big)$ is the probability of observing outcome `0' when $M_0$ acts on half of two singlets, and the channel corresponding to the reference measurement is defined according to $\overline{M_0}[\cdot] =  \t{tr}_{\sA}\left(\zeta_0\,\proj{\phi_{00}} \otimes \Id_\sB [\cdot]\right)$.
\\

Let $\varrho_0$ be the following conditional state
\be
\varrho_0 = \frac{ M_0( \rho )}{\t{Tr} \,M_0( \rho )} \in L(\cH_B^{(1)}\otimes \cH_B^{(2)}).
\ee
From Eq. (10) of the main text, we have 
\be
\label{rho2rho3}
F\left( (\Lambda_B^{(1)} \otimes \Lambda_B^{(2)})[\varrho_0], \ket{\phi_{00}}\right)\geq F^o_0= \sqrt{1 -\frac{1}{2}\cdot\frac{2\sqrt{2}-\beta_k}{2\sqrt 2 -\beta^*}}
\ee  
where $\beta_0$ is the CHSH value observed between $B^{(1)}$ and $B^{(2)}$ when the supposedly Bell state measurement outcomes 0 when applied on $\rho$. At the same time, the self-testing of the state produced by the source in step II ensures that
\begin{equation}
F(\widetilde{\Lambda}_A^{(1)}\otimes \widetilde{\Lambda}_A^{(2)} \otimes \Lambda_B^{(1)} \otimes \Lambda_B^{(2)} [\rho],|\phi_{00}\rangle^{\otimes 2}) \geq F^i.   
\end{equation}
with the same isometries $\Lambda_B^{(i)}$ on Bob's side, and some isometries $\widetilde{\Lambda}_A^{(i)}$ on Alice's. As shown in the proof of Proposition 4 in \cite{Sekatski_18}, the relation above guarantees the existence of new isometries on Alice's side, $\Lambda_A^{(i)}$, such that
\begin{equation}
\label{eq_fidinitial}    
F( \Lambda_B^{(1)} \otimes \Lambda_B^{(2)} [\rho],\Lambda_A^{(1)}\otimes \Lambda_A^{(2)} [\proj{\phi_{00}}^{\otimes 2}\!\otimes \rho_\sA])\geq F^i,  
\end{equation}
where $\rho_\sA$ is the state of all the auxiliary degrees of freedom at side $\sA$ distributed by the source.\\

We now have the basic ingredients that will give rise to our bound on the quality of probabilistic and partial Bell state measurements. To put them together, we add one more tool to the ones already present in~\cite{Sekatski_18} to deal with probabilistic channels.

\begin{lemma}
Given two quantum states $\rho$ and $\sigma$ with Uhlmann fidelity $F(\rho,\sigma)=F$, and a probabilistic quantum channel $\cE_0$ with $\tr {\cE_0[\rho]}=p_0$ and $\tr{ \cE_0[\sigma]}=q_0$. The fidelity of the final states $\rho_0 = \frac{\cE_0[\rho]}{p_0}$ and $\sigma_0= \frac{ \cE_0[\sigma]}{q_0}$ satisfies
\be\label{eq:proc prob}
\sqrt{q_0 p_0} F\left(\rho_0,\sigma_0 \right) \geq F - \sqrt{(1-p_0)(1-q_0)}.
\ee
\end{lemma}
\begin{proof}
A probabilistic quantum channel $\cE_0$ can be completed with the branch $\cE_\emptyset$ in order to form a quantum instrument $\cE= \{ \cE_0, \cE_\emptyset\}$ (a trace preserving quantum channel with an output a label "$0$" or "$\emptyset$"). The states after the action of $\cE$ read
\begin{align}
\cE(\rho) &= p_0 \rho_0 \oplus (1-p_0) \rho_\emptyset \\ \cE(\sigma) &= q_0 \sigma_0 \oplus (1-q_0) \sigma_\emptyset,
\end{align}
with $\tau_\emptyset = \frac{\cE_\emptyset[\tau]}{\tr{\cE_\emptyset[\tau]}}$.

Using the processing inequality of fidelity (the fact that it can not decrease by post-processing) we get
\begin{align}
\nonumber
    F &\leq F(\cE(\rho), \cE(\sigma)) \\
    \nonumber
    &= \sqrt{q_0 p_0} F(\rho_0,\sigma_0) + \sqrt{(1-p_0)(1-q_0)}
    F(\rho_\emptyset,\sigma_\emptyset).
\end{align}
Rearranging the terms and using $ F(\rho_\emptyset,\sigma_\emptyset)\leq 1$  we obtain the desired relation
\be
\sqrt{q_0 p_0}  F\left(\rho_0,\sigma_0 \right)\geq F - \sqrt{(1-p_0)(1-q_0)}.
\ee
\end{proof}

Applying inequality~\eqref{eq:proc prob} to the states $\Lambda_B^{(1)} \otimes \Lambda_B^{(2)} [\rho]$,  $\Lambda_A^{(1)}\otimes \Lambda_A^{(2)} [\proj{\phi_{00}}^{\otimes 2}\otimes \rho_\sA]$, and the supposed probabilistic Bell state measurement $M_0$ gives the following bound
\begin{align}\label{eq: F cond}
\sqrt{p_0 \frac{\zeta_0}{4}} F_c \geq F^i-\sqrt{(1-p_0)(1-\zeta_0/4)}
\end{align}
for the conditional fidelity
\begin{align}
F_c = F&\left(\frac{1}{p_0} \Lambda_B^{(1)} \otimes \Lambda_B^{(2)}\otimes M_0 [\rho],\right. \\
&\left.\frac{4}{\zeta_0} M_0\circ(\Lambda_A^{(1)}\otimes \Lambda_A^{(2)}) [\proj{\phi_{00}}^{\otimes 2}\!\!\otimes \! \rho_\sA]\right).
\end{align}\\

The inequality \eqref{eq: F cond} simultaneously lower bounds the rate $\zeta_0$ and the fidelity $F_c$ from the values of $p_0$ and $F_i$ that can be observed experimentally. But $\zeta_0$ is not directly observable. Hence, this inequality is only interesting when it sets a non-trivial bound on $F_c$ for all values of $\zeta_0$ compatible with the experiment. In particular, this means that the rhs of \eqref{eq: F cond} has to remain positive for arbitrarily small values of the rate $\zeta_0 \to 
 0$ (otherwise the observed values $F_i$ and $p_0$ are compatible with $\zeta_0 = 0$ and any value of $F_c$). Thus, 
we are interested in the regime
\be
\left(F^i\right)^2 + p_0 > 1.
\ee
In this regime, Eq.~\eqref{eq: F cond} can be simplified according to
\be\label{eq: bound Fc}
F_c \geq \frac{ F^i-\sqrt{(1-p_0)(1-\zeta_0/4)}}{\sqrt{p_0 \zeta_0/4}}\geq \sqrt{\frac{p_0 +\left(F^i\right)^2-1}{p_0}},
\ee
where for the last inequality we simply minimize the expression over $\zeta_0$: we find the unique value where the derivative with respect to $\zeta_0$ is zero and ensure that it is a minimum by verifying the positivity of the second derivative at that point.\\

From the same Eq.~\eqref{eq: F cond} and $F_c\leq 1$, we also obtain a bound on the rate:
\be\label{eq: rate lb}
\zeta_0 \geq 4 \left(\sqrt{p_0 \left(F^i\right)^2} - \sqrt{(1-p_0)\left(1-\left(F^i\right)^2\right)}\right)^2.
\ee
We remark that the bounds on the rate and on the conditional fidelity are not independent, i.e. for the minimal value of $\zeta_0$ above, the value of $F_c$ can be shown to be larger than the bound in Eq.~\eqref{eq: bound Fc}. For more refined statements on the pair $(F_c,\, \zeta_0)$ one can directly work with Eq.~\eqref{eq: F cond}.\\

Using the equivalent of the triangle inequality for the Uhlmann fidelity, we combine Eq.~\eqref{rho2rho3} and Eq.~\eqref{eq: bound Fc} to obtain a lower bound on the quality of the partial Bell state measurement
\begin{align}
\mathcal{F}_\text{cond} \left(M_0, \overline{ M_0} \,\right)
\geq \cos\Bigg(&\arccos{(F^o_0)}
\nonumber
\\
+&\arccos \left(\sqrt{\frac{p_0  + (F^i)^2 -1}{p_0}}\right)\Bigg)
\nonumber.
\end{align}
This fidelity and the rate $\zeta_0$ are plotted in Figs. \ref{fig: 5} and \ref{fig: 6} as a function of $\beta_0$ assuming $\delta=\beta_0/(2\sqrt{2})$ for $p_0=1/4,$ $p_0=1/10$ and $p_0=1/100.$ 

\begin{figure}
\begin{center}
\includegraphics[width=1\columnwidth]{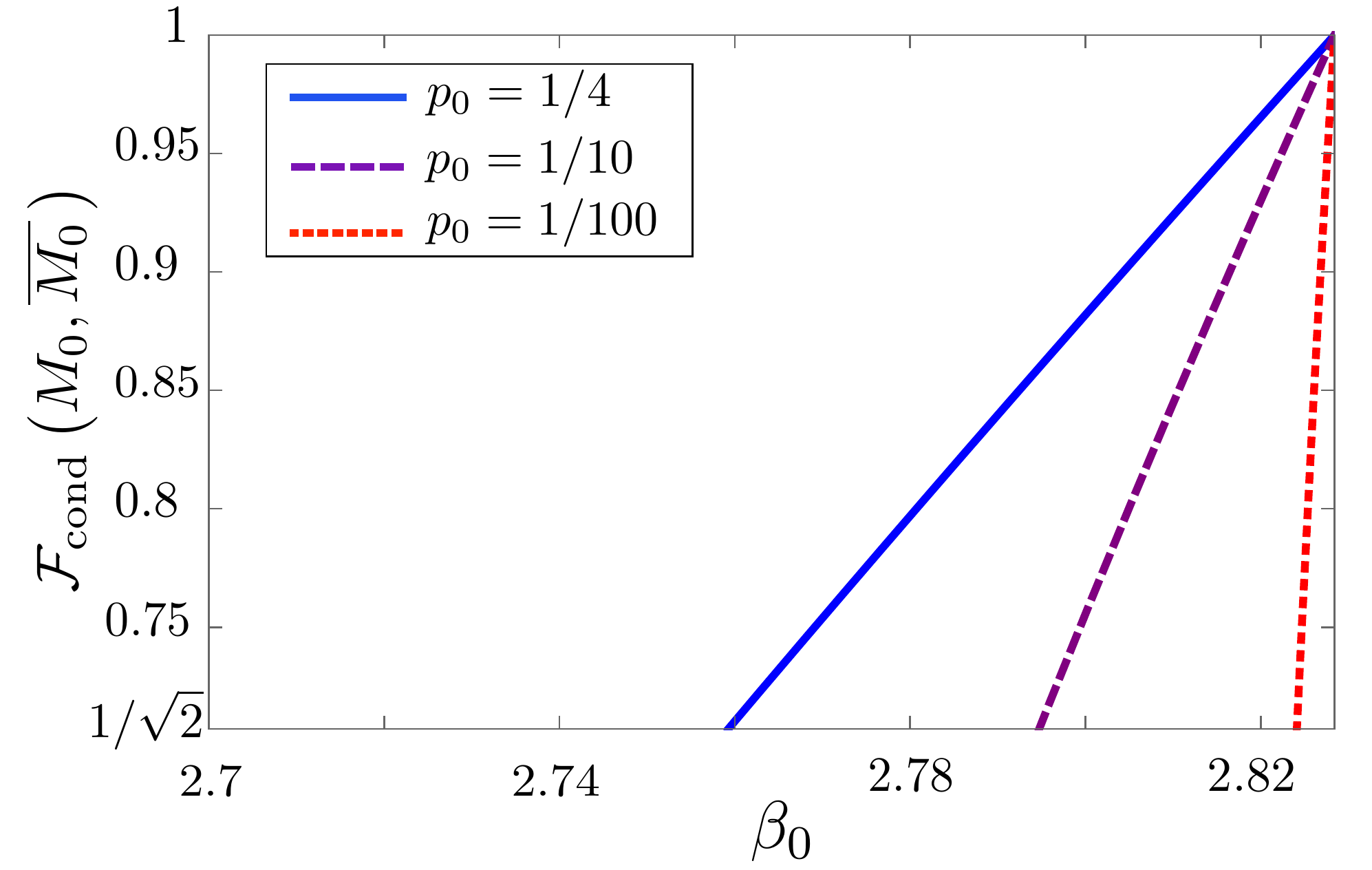}
\caption{Fidelity of a probabilistic/partial Bell state measurement as a function of the CHSH value $\beta_0$ obtained by parties $B^{(1)}$ and $B^{(2)}$ on the post-measurement state associated with outcome $0$. We consider the case where the certification of the source led to $\delta=\beta_0/{2\sqrt{2}}$ (Step II). The full line, dashed line and dashed with smaller dashed line correspond to the heralding probability $p_0=1/4,$ $p_0=0.1$ and $p_0=0.01$ respectively. }
\label{fig: 5}
\end{center}
\end{figure}

\begin{figure}
\begin{center}
\includegraphics[width=1\columnwidth]{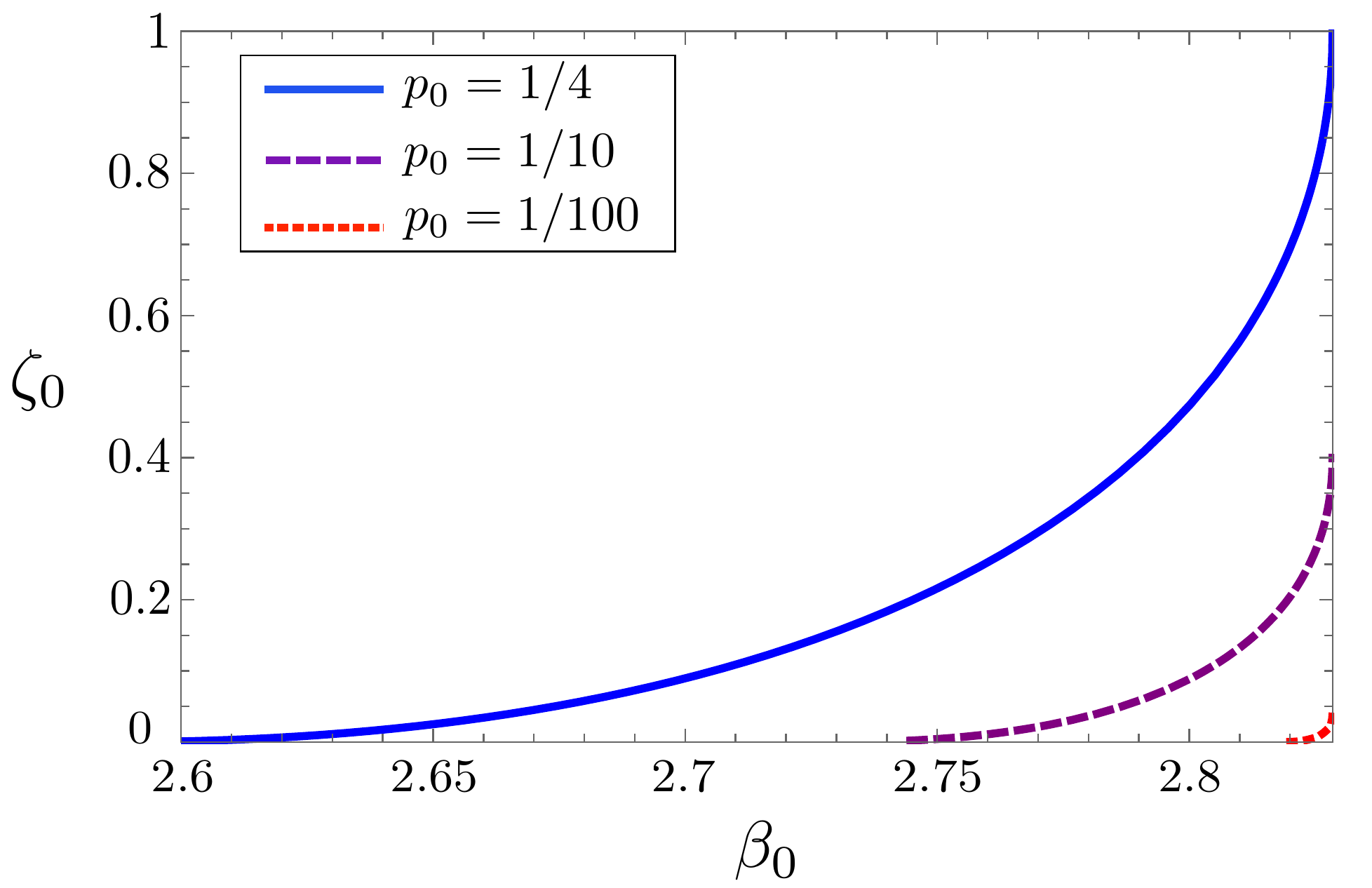}
\caption{Bound on $\zeta_0$ as a function of the CHSH value $\beta_0$ obtained by parties $B^{(1)}$ and $B^{(2)}$ on the post-measurement state associated with outcome $0$ of the BSM. We recall that $\zeta_0/4$ is the probability that the actual Bell state measurement supplemented with the proper injection maps produce the outcome `0' when operating on two maximally entangled two-qubit states. We consider the case where the certification of the source led to $\delta=\beta_0/(2\sqrt{2})$ (Step II). The full line, dashed line and dashed with smaller dashed line correspond to the observed heralding probability $p_0=1/4,$ $p_0=0.1$ and $p_0=0.01$ respectively.}
\label{fig: 6}
\end{center}
\end{figure}

\bibliography{BSM}

\end{document}